\documentclass[runningheads]{llncs}
\usepackage{amssymb}
\usepackage{graphicx}
\usepackage{epstopdf}
\usepackage{multirow}
\usepackage{bm}

\begin{document}
\title{An Improved Approximation for Packing Big Two-Bar Charts\thanks{The research
is carried out within the framework of the state contract of the
Sobolev Institute of Mathematics (project 0314--2019--0014).}}
\author{Adil Erzin\inst{1}\orcidID{0000-0002-2183-523X} \and
Vladimir Shenmaier\inst{1}\orcidID{0000-0002-4692-1994}}
\authorrunning{A. Erzin, V. Shenmaier}
\institute{Sobolev Institute of Mathematics, SB RAS, Novosibirsk 630090, Russia\\
\email adilerzin@math.nsc.ru, shenmaier@mail.ru}
\maketitle              
\begin{abstract}
Recently, we presented a new Two-Bar Charts Packing Problem (2-BCPP), in which it is necessary to pack two-bar charts (2-BCs) in a unit-height strip of minimum length. The problem is a generalization of the Bin Packing Problem and 2-D Vector Packing Problem. Earlier, we have proposed several polynomial approximation algorithms. In particular, when each 2-BC has at least one bar of height more than 1/2, we have proposed a 3/2--approximation polynomial algorithm. This paper proposes an $O(n^3)$--time 16/11--approximation algorithm for packing 2-BCs when at least one bar of each BC has a height not less than 1/2 and an $O(n^{2.5})$--time 5/4--approximation algorithm for packing non-increasing or non-decreasing 2-BCs when each 2-BC has at least one bar which height is more than 1/2, where $n$ is the number of 2-BCs.
\keywords{Two-Bar Charts \and Strip Packing \and MaxTSP \and Approximation}
\end{abstract}
\section{Introduction}
Where the problem in question came from can be found in our papers \cite{Erzin20_1,Erzin20_2}. In its refined form, it is formulated as follows. We have a set of bar charts (BCs) consisting of two bars each. Any bar has a length equal to 1, and its height does not exceed 1. Let us denote such charts as 2-BCs. It is required to find a \emph{feasible} min-length packing of all 2-BCs in a unit-height strip. If we divide the strip into equal unit-length cells, then the packing length is the number of cells in which there is at least one bar. In a feasible packing, each BC's bars do not change order and must occupy adjacent cells, but they can move vertically independently of each other.

This problem was first examined in \cite{Erzin20_2}, where we described similar problems that have been well studied. Similar problems that were studied reasonably well are the bin packing problem (BPP) \cite{Baker85,Dosa07,Johnson73,Johnson85,Li97,Yue91,Yue95}, the strip packing problem (SPP)\cite{Baker80,Coffman80,Harren09,Harren14,Schiermeyer94,Steinberg97}, and the two-dimensional vector packing problem (2-DVPP) \cite{Bansal16,Christensen17,Kellerer03,Wei20}.

In the BPP, a set of items $L$, each item's size, and a set of identical containers (bins) are given. All items must be placed in a minimum number of bins. BPP is a strongly NP-hard problem. However, many approximate algorithms have been proposed for it. One of the well-known algorithms is First Fit Decreasing (FFD). As part of this algorithm, objects are numbered in non-increasing order. All items are scanned in order, and the current item is placed in the first suitable bin.  In 1973 was proved that the FFD algorithm uses no more than $11/9\ OPT(L)+4$ containers \cite{Johnson73}, where $OPT(L)$ is the minimal number of bins to pack the items from the set $L$. Then in 1985, the additive constant was reduced to 3 \cite{Baker85}, in 1991 it was reduced to 1 \cite{Yue91}, in 1997 it was reduced to 7/9 \cite{Li97}, and finally in 2007 was found the tight boundary of the additive constant equals 6/9 \cite{Dosa07}. A Modified First Fit Decreasing (MFFD) algorithm improves FFD. It was shown that $MFFD(L)\leq 71/60\ OPT(L)+31/6$ \cite{Johnson85}. Then the result was improved to $MFFD(L)\leq 71/60\ OPT(L)+1$ \cite{Yue95}.

In the Strip Packing Problem (SPP) for each rectangle $i\in L$, we know its length and height. It is required to pack all rectangles (without rotation) in a minimum length strip. The Bottom-Left algorithm \cite{Baker80} arranges rectangles in descending order of height and yields a 3--approximate solution. In 1980 was proposed algorithms with ratio 2.7 \cite{Coffman80}. Sleator \cite{Sleator80} proposed a 2.5--approximate algorithm, and this ratio was reduced by Schiermeyer \cite{Schiermeyer94} and Steinberg \cite{Steinberg97} to 2. The smallest estimate for the ratio known to date is $(5/3+\varepsilon)OPT(L)$, for any $\varepsilon > 0$ \cite{Harren14}.

The 2-DVPP is a generalization of the BPP and a particular case of 2-BCPP. It considers two attributes for each item and bin. The objective is to minimize the number of containers used. In \cite{Kellerer03} was presented a 2--approximate algorithm for 2-DVPP vector packing. In \cite{Christensen17}, one can find a survey of approximation algorithms. The best one yields a $(3/2+\varepsilon)$--approximate solution, for any $\varepsilon >0$ \cite{Bansal16}.

In \cite{Erzin20_2}, we proposed an $O(n^2)$--time algorithm to build a packing for $n$ 2-BCs which length is at most $2\ OPT+1$, where $OPT$ is the minimum packing length for 2-BCPP. Then in \cite{Erzin20_3}, we presented the polynomial algorithms to solve the particular cases of the 2-BCPP when all BCs are ``big'' (at least one bar's height of each BC is more than 1/2). For the case of big non-increasing or non-decreasing BCs (when either all BCs have the first bar not less than the second, or all BCs have the second bar not less than the first), an $O(n^{3.5})$--time 3/2--approximate algorithm was proposed. For arbitrary big BCs, the $O(n^4)$--time 3/2--approximate algorithm has been proposed. The indicated time complexity characterizes the developed algorithms' time execution based on the construction of $O(n)$ matchings. To achieve the specified accuracy, one can construct only one (first) matching. Therefore, the complexity of getting the specified accuracy is $O(n^{2.5})$ and $O(n^3)$, respectively.

\subsection{Our contribution}
This paper updates the estimates for the packing length of big 2-BCs, keeping the time complexity. First, we give a $5/4$--approximation $O(n^{2.5})$--time algorithm for the version of 2-BCPP which we call 2-BCPP$|$1 and which contains all the instances of 2-BCPP with big non-increasing or non-decreasing 2-BCs. In 2-BCPP$|$1, we are given arbitrary (not necessarily big and not necessarily non-increasing or non-decreasing) 2-BCs, and the goal is to find a min-length packing in which two neighboring 2-BCs intersect by at most one cell on the strip.
The proposed algorithm is based on an approximation-preserving reduction of 2-BCPP$|$1 to the maximum traveling salesman problem and using the result of \cite{Paluch} for the latter.
It should be noted that MaxTSP is one of the most intensively researched optimization problems but has very few natural, real-life applications.
2-BCPP$|$1 can be considered as one of such applications.

The second result is a $16/11$--approximation $O(n^3)$--time algorithm for 2-BCPP with big BCs.
This estimate is valid for the case of ``non-strictly big'' BCs, in which at least one bar is of height at least 1/2.
In obtaining this estimate, we use the known algorithm for finding a max-cardinality matching \cite{Gabow} and the proposed approximation algorithm for 2-BCPP$|$1.

\vspace{0.5cm}
The rest of the paper is organized as follows.
Section 2 provides a statement of the 2-BCPP and necessary definitions.
In Section 3, we describe a $5/4$--approximation algorithm for 2-BCPP$|$1.
In Section 4, we describe a $16/11$--approximation algorithm for the case of 2-BCPP in which all the 2-BCs are non-strictly big.
Section 5 concludes the paper.

\section{Formulation of the problem}
Let a semi-infinite strip of unit height be given on a plane in the first quadrant, the lower boundary of which coincides with the abscissa. A set $S$, $|S|=n$ of 2-BCs is also given. Each chart, $i\in S$, consists of two unit-length bars. The height of the first (left) bar is $a_i\in (0,1]$ and of the second (right) $b_i\in (0,1]$. Let us divide the strip into equal unit-length and unit-height rectangles (cells), starting from the strip's origin, and number them with integers $1,2,\ldots$.

\begin{definition}
\emph{Packing} is a function $p:S\rightarrow \mathbb{Z}^+$, which associates with each BC $i$ the cell number of the strip $p(i)$ into which the first bar of BC $i$ falls and the sum of the bar's heights that fall into any cell does not exceed 1.
\end{definition}

As a result of packing $p$, bars from 2-BC $i$ occupy the cells $p(i)$ and $p(i)+1$.

\begin{definition}
The packing \emph{length} $L(p)$ is the number of strip cells in which at least one bar falls.
\end{definition}

We assume that any packing $p$ begins from the first cell, and in each cell $1,\ldots,L(p$), there is at least one bar. If this is not the case, then the whole packing or a part of it can be moved to the left.

In \cite{Erzin20_2}, we formulated 2-BCPP in the form of BLP, which we do not need in this paper. Here the problem can be formulated as follows.

\vspace{0.5cm}
\textbf{2-BCPP: Given an $\bm n$-element set $\bm S$ of $\bm2$-BCs, it is required to construct a packing of $S$ into a strip of minimum length, i.e., using the strip's minimum number of cells.}
\vspace{0.5cm}

The 2-BCPP is strongly NP-hard as a generalization of the BPP \cite{Johnson73}. Moreover, the problem is $(3/2-\varepsilon)$--inapproximable unless P=NP \cite{Vazirani01}. In \cite{Erzin20_2}, we proposed an $O(n^2)$--time algorithm, which packs the 2-BCs in the strip of length at most $2\ OPT+1$, where $OPT$ is the minimum packing length. In \cite{Erzin20_3}, we proposed two packing algorithms based on sequential matchings. Using only the first matching, one can construct a 3/2--approximate solution with time complexity $O(n^3)$ for the case when all BCs are big and with $O(n^{2.5})$ time complexity when additionally the BCs are non-increasing or non-decreasing.

This paper proposes two new packing algorithms based on matching and approximation solutions to the max-weight Hamiltonian tour in the complete digraph with arcs' weight 0 or 1. If all 2-BCs are non-strictly big, one algorithm constructs a 16/11--approximate solution with time complexity $O(n^3)$. If all 2-BCs are big non-increasing or non-decreasing, then the other algorithm constructs a 5/4--approximate solution with time complexity $O(n^{2.5})$.

\begin{definition}
Two 2-BCs form a $t$-union if they can be placed in the strip's $4-t$ cells.
\end{definition}

In what follows, we will need a problem of constructing a max-weight Hamiltonian tour in a complete digraph with arc's weights 0 and 1. Let us denote this problem as MaxATSP(0,1).

\section{A $5/4$--approximation for 2-BCPP limited to 1-unions}
In this section, we describe a $5/4$--approximation algorithm for the version of 2-BCPP, where we are allowed to use only 0- and 1-unions:

\vspace{0.5cm}
\textbf{$\bm2$-BCPP$\bm{|1}$: Given an $\bm n$-element set $\bm S$ of $\bm2$-BCs, it is required to construct a min-length packing of $\bm S$ into a strip in which every pair of successive BCs forms a $\bm t$-union, where $\bm{t\le 1}$.}
\vspace{0.5cm}

In this version, the left and the right bars in each BC $i\in S$ may have arbitrary values $a_i,b_i\in(0,1]$. It is easy to see that 2-BCPP$|$1 contains all the instances of 2-BCPP with non-increasing (or non-decreasing) big 2-BCs.

The suggested algorithm will be used not only for $2$-BCPP$|1$ but also for finding approximate solutions of $2$-BCPP with big $2$-BCs (see Sect.~4). This algorithm is based on a simple approximation-preserving reduction of $2$-BCPP$|1$ to the maximum traveling salesman problem and using known algorithmic results for the latter. It seems quite impressive in learning $2$-BCPP but also in learning MaxTSP since we give a natural, real-life application for MaxTSP, which is one of the most studied optimization problems that have almost no practical applications.

\subsection{A reduction of $2$-BCPP$|1$ to MaxATSP$(0,1)$}
Consider the complete weighted directed graph $G_1(S)=(S,S^2)$ in which the weight of each arc $(i,j)\in S^2$ is defined as
$$\left\{
\begin{array}{cl}
 1 &\mbox{if $b_i+a_j\le 1$,}\\
 0 &\mbox{otherwise.}
\end{array}
\right.$$
Informally speaking, the weights of the edges in $G_1(S)$ describe which pairs of BCs can form a $1$-union and which cannot. Note that, in the general case, these weights are asymmetric since $b_i+a_j$ may differ from $b_j+a_i$.

Denote by ${\cal P}_1(S)$ the set of $S$ packings which consist of $0$- and $1$-unions. Then, for each packing $P\in{\cal P}_1(S)$, define the Hamiltonian cycle $H(P)$ in $G_1(S)$ with the sequence of vertices $i_1,\dots,i_n,i_1$, where $i_1,\dots,i_n$ is the sequence of $2$-BCs in $P$ in order left to right.
Obviously, we have
\begin{equation}\label{eq1}
w(H(P))\ge k_1(P),
\end{equation}
where $w(.)$ is the total weight of a Hamiltonian cycle and $k_1(.)$ denotes the number of $1$-unions in a packing.

Now, let $H$ be an arbitrary Hamiltonian cycle in $G_1(S)$.
If $w(H)<n$, select a BC $i\in S$ such that the arc in $H$ incoming to $i$ is of zero weight; otherwise, select any chart $i\in S$.
Let $i_1,\dots,i_n,i_1$ be the sequence of vertices in $H$ starting with $i_1=i$ and define the packing $P(H)$ whose sequence of charts in order left to right is $i_1,\dots,i_n$ and each pair $(i_t,i_{t+1})$, $1\le t<n$, forms a $1$-union whenever it is possible, i.e., when $b_{i_t}+a_{i_{t+1}}\le 1$.
Then it is easy to see that
\begin{equation}\label{eq2}
k_1(P(H))=\min\{w(H),n-1\}.
\end{equation}

Immediate corollaries of (\ref{eq1}),\,(\ref{eq2}) are the following simple statements, which give the desired reduction to MaxATSP(0,1).

\begin{lemma}\label{lem1}
If $H^*$ is a max-weight Hamiltonian cycle in $G_1(S)$, then $P(H^*)$ is a min-length packing in ${\cal P}_1(S)$.
\end{lemma}

\begin{proof}
Indeed, by (\ref{eq2}), we have $k_1(P(H^*))=\min\{w(H^*),n-1\}$.
Suppose that $k_1(P)>\min\{w(H^*),n-1\}$ for some $P\in{\cal P}_1(S)$.
Then, since the number of $1$-unions in any packing is at most $n-1$, we have $k_1(P)>w(H^*)$.
So, by (\ref{eq1}), we obtain that $w(H(P))\ge k_1(P)>w(H^*)$, which contradicts the choice of $H^*$.
The lemma is proved.
\hfill$\Box$
\end{proof}

\begin{lemma}\label{lem2}
Suppose that $H$ is an $\alpha$--approximate solution to MaxATSP(0,1) in $G_1(S)$ for some $\alpha\in(0,1)$. Then the number of
1-unions in the packing $P(H)$ is at least $\alpha$ of that in an optimum packing in ${\cal P}_1(S)$.
\end{lemma}

\begin{proof}
By Lemma~1, if $H^*$ is a max-weight Hamiltonian cycle in $G_1(S)$, then $P(H^*)$ is a min-length packing in ${\cal P}_1(S)$, so the number of $1$-unions in optimum packings is $k_1(P(H^*))$.
On the other hand, by (\ref{eq2}), we have
$$k_1(P(H))=\min\{w(H),n-1\}\ge\alpha\min\{w(H^*),n-1\}=\alpha k_1(P(H^*)).$$
The lemma is proved.
\hfill$\Box$
\end{proof}

By Lemma~\ref{lem2}, if $H$ is an $\alpha$--approximate solution to MaxATSP(0,1) in $G_1(S)$ and $k_1^*$ is the number of $1$-unions in an optimum packing in ${\cal P}_1(S)$, then the approximation ratio of the packing $P(H)$ is at most
\begin{equation}\label{eq3}
\frac{2n-\alpha k_1^*}{2n-k_1^*}\le\frac{2n-\alpha(n-1)}{2n-(n-1)}=\frac{(2-\alpha)\,n+\alpha}{n+1}<2-\alpha.
\end{equation}
Thus, we obtain an approximation-preserving reduction of $2$-BCPP$|1$ to the maximum traveling salesman problem with asymmetric weights $0$ and $1$, which is usually denoted MaxATSP$(0,1)$.
This reduction transforms $\alpha$--approximate solutions of the corresponding instances of MaxATSP$(0,1)$ to $(2-\alpha)$--approximate solutions of $2$-BCPP$|1$.

Moreover, it is easy to construct such a reduction to the partial case of MaxATSP$(0,1)$ where the input graphs are of even size.
Indeed, if $n$ is odd, we will consider the set $S'=S\cup\{\tau\}$, where $\tau$ is the dummy chart with bars $a_\tau=1$ and $b_\tau=1$.

\begin{lemma}\label{lem3}
Suppose that $H$ is an $\alpha$--approximate solution to MaxATSP(0,1) in $G_1(S')$ for some $\alpha\in(0,1)$ and $P$ is the packing obtained from $P(H)$ by removing $\tau$.
Then the number of $1$-unions in the packing $P$ is at least $\alpha$ of that in an optimum packing in ${\cal P}_1(S)$.
\end{lemma}

\begin{proof}
Let $H^*$ be an optimum solution of MaxATSP(0,1) on $G_1(S')$.
Then, by (\ref{eq2}), we have $k_1(P(H))=\min\{w(H),n\}\ge\min\{\alpha w(H^*),n\}$.
But the BC $\tau$ can not be in any $1$-union, so $k_1(P)=k_1(P(H))$.
At the same time, if $P^*$ is an optimum packing in ${\cal P}_1(S)$ and the sequence of charts in $P^*$ in order left to right is $i_1,\dots,i_n$, then the weight of the Hamiltonian cycle $H_\tau(P^*)=(i_1,\dots,i_n,\tau,i_1)$ in $G_1(S')$ is exactly $k_1(P^*)$.
So we have
$$k_1(P)=k_1(P(H))\ge\min\{\alpha w(H^*),n\}\ge\min\{\alpha w(H_\tau(P^*)),n\}=\alpha k_1(P^*).$$
The lemma is proved.
\hfill$\Box$
\end{proof}

Thus, if $H$ is an $\alpha$--approximate solution to MaxATSP(0,1) in $G_1(S')$ and, as before, $k_1^*$ is the number of $1$-unions in an optimum packing in ${\cal P}_1(S)$, then we obtain a solution of $2$-BCPP$|1$ on the set $S$ with approximation ratio bounded by the same expressions as in (\ref{eq3}).

\subsection{An algorithm for $2$-BCPP$|1$}
It remains to recall the known algorithmic results for MaxATSP$(0,1)$.
Currently best approximations for MaxATSP$(0,1)$ are the $3/4$--ap\-p\-ro\-xi\-ma\-ti\-on LP-based algorithm of Bl\"{a}ser \cite{Blaser} and the $3/4$--approximation combinatorial algorithm of Paluch \cite{Paluch}.
The running time of the latter algorithm is $O(n^{2.5})$ if $n$ is even and $O(n^{3.5})$ if $n$ is odd.
We suggest using the above reduction to MaxTSP with an even number of vertices and applying Paluch's algorithm.
The resulting algorithm for $2$-BCPP$|1$ can be described as follows:\bigskip

\noindent\textbf{Algorithm~${\cal A}_1$.}\smallskip

\noindent\emph{Input:} a set $S$ of $n$ 2-BCs.
\emph{Output:} a packing $P\in{\cal P}_1(S)$.\medskip

\noindent\emph{Step~{\rm 1}.}
If $n$ is even, construct the graph $G=G_1(S)$; otherwise, construct the graph $G=G_1(S')$, where $S'=S\cup\{\tau\}$, $a_\tau=b_\tau=1$.\medskip

\noindent\emph{Step~{\rm 2}.}
By using the algorithm from \cite{Paluch}, find a $3/4$--approximate solution $H$ to MaxATSP(0,1) in $G$.\medskip

\noindent\emph{Step~{\rm 3}.}
If $n$ is even, return $P=P(H)$;
otherwise, return the packing $P$ obtained from $P(H)$ by removing $\tau$.\bigskip

By Lemmas \ref{lem2}, \ref{lem3} and estimate (\ref{eq3}), the approximation ratio of the packing returned by Algorithm~${\cal A}_1$ is less than $2-3/4=5/4$.
So we prove

\begin{theorem}\label{th5_4}
Algorithm~${\cal A}_1$ finds a $5/4$--approximate solution to $2$-BCPP$|1$ in time $O(n^{2.5})$.
\end{theorem}

\section{A $16/11$--approximation for $2$-BCPP with big charts}
This section describes a $16/11$--approximation algorithm for the case of $2$-BCPP in which all $2$-BCs are non-strictly big.
Here, we assume that a chart $i$ with bars $a_i,b_i\in(0,1]$ is \emph{non-strictly big} if $\max\{a_i,b_i\}\ge 1/2$, i.e., the case when $\max\{a_i,b_i\}=1/2$ is admissible.

\vspace{0.5cm}
\textbf{$\bm2$-BCPP$\bm{|}$big: Given an $\bm n$-element set $\bm S$ of non-strictly big $\bm2$-BCs, it is required to construct a min-length packing of $\bm S$.}
\vspace{0.5cm}

Let us make some simple observations.
First, it is easy to see that, since all the charts in $S$ are non-strictly big, then any two pairs of charts $\{i_1,i_2\}$ and $\{i_3,i_4\}$ which form $2$-unions in any packing of $S$ are disjoint.
On the other hand, any set of disjoint pairs of charts forming $2$-unions gives a feasible solution of $2$-BCPP${|}$big.
In particular, we can get such a solution by finding a max-cardinality matching $M^*$ in the graph $G_2(S)$ whose vertices are the charts of $S$ and the edges are the unordered pairs $i,j\in S$ admitting a $2$-union, i.e., for which $a_i+a_j\le 1$ and $b_i+b_j\le 1$.
Let $P(M^*)$ be the packing of $S$, which consists of $0$- and $2$-unions and all $2$-unions formed by endpoints of the edges in $M^*$.

The second observation is that any feasible solution of $2$-BCPP$|1$ on $S$ is also that of $2$-BCPP${|}$big.
On the other hand, any feasible solution to the latter problem can be easily transformed to $2$-BCPP$|1$ as follows.
If two charts $i$ and $j$ form a $2$-union and $a_i+b_j\ge a_j+b_i$, then we shift the BC $j$ and all the charts lying to the right of it one cell right. As a result, we get one cell with content height $b_i+a_j\le(a_i+b_j+a_j+b_i)/2\le 1$, while the content of the other cells does not increase. So we get a feasible packing where the pair $(i,j)$ of BCs forms a $1$-union. Denote the packing constructed by the described processing of all $2$-unions in $P$ as $\Gamma(P)$.
Then $\Gamma(P)\in{\cal P}_1(S)$ and the number of $1$-unions in $\Gamma(P)$ is exactly the total number of $1$- and $2$-unions in $P$.

If an optimum solution of $2$-BCPP${|}$big contains a small number of $1$-unions and a significant number of $2$-unions, this packing is not much shorter than $P(M^*)$. If, on the contrary, an optimum solution of $2$-BCPP${|}$big contains a significant number of $1$-unions and a small number of $2$-unions, then it is not much shorter than an optimum solution of $2$-BCPP$|1$. So the best of $P(M^*)$ and an optimum of $2$-BCPP$|1$ may be a relatively right solution to $2$-BCPP${|}$big. Based on this hypothesis, we suggest the following algorithm:\bigskip

\noindent\textbf{Algorithm~${\cal A}_2$.}\smallskip

\noindent\emph{Input:} a set $S$ of $n$ big 2-BCs.
\emph{Output:} a packing $P$ of $S$.\medskip

\noindent\emph{Step~{\rm 1}.}
By using the algorithm from \cite{Gabow}, construct a max-cardinality matching $M^*$ in $G_2(S)$.\medskip

\noindent\emph{Step~{\rm 2}.}
By using Algorithm~${\cal A}_1$, find an approximate solution $P_1$ to $2$-BCPP$|1$ on $S$.\medskip

\noindent\emph{Step~{\rm 3}.}
If the length of $P(M^*)$ is less than that of $P_1$, return $P=P(M^*)$;
otherwise, return $P=P_1$.\smallskip

\begin{theorem}\label{th16_11}
Algorithm~${\cal A}_2$ finds a $16/11$--approximate solution to $2$-BCPP$|$big in time $O(n^3)$.
\end{theorem}

\begin{proof}
Let $P^*$ be a min-length packing of $S$. Denote by $k_1$ and $k_2$ the numbers of $1$- and $2$-unions in $P^*$, respectively.
Then, since the $2$-unions in $P^*$ form a matching in $G_2(S)$, the cardinality of $M^*$ is at least $k_2$.
Therefore, the length of $P(M^*)$ is at most $2n-2k_2$.

Next, we estimate the length of $P_1$.
By the construction of the packing $\Gamma(P^*)$, we have $k_1(\Gamma(P^*))=k_1+k_2$.
By the description of Algorithm~${\cal A}_1$ and Lemmas \ref{lem2} and~\ref{lem3}, the number of $1$-unions in $P_1$ is at least $3/4$ of that in any packing in ${\cal P}_1(S)$.
In particular, we have
$$k_1(P_1)\ge(3/4)\,k_1(\Gamma(P^*))=(3/4)(k_1+k_2).$$
So the length of the packing $P_1$ is at most $2n-(3/4)(k_1+k_2)$.

At the same time, it is easy to see that the length of $P^*$ is $2n-k_1-2k_2$.
Moreover, as shown in \cite{Erzin20_3}, this length is at least $n$.
It follows that the approximation ratio of the solution returned by Algorithm~${\cal A}_2$ is at most
\begin{equation}\label{16_11_1}
\frac{\min\{2n-2k_2,\ 2n-(3/4)(k_1+k_2)\}}{2n-k_1-2k_2},
\end{equation}
where $2n-k_1-2k_2\ge n$ or, equivalently, where $k_1\le n-2k_2$.
Obviously, the maximum value of expression (\ref{16_11_1}) is attained at the maximum possible value of $k_1$, i.e., when $k_1=n-2k_2$.
Therefore, this expression is bounded by
\begin{equation}\label{16_11_2}
\frac{\min\{2n-2k_2,\ 2n-(3/4)(n-k_2)\}}{n}.
\end{equation}
Next, the values $2n-2k_2$ and $2n-(3/4)(n-k_2)$ are decreasing and increasing functions of $k_2$ respectively, while the denominator in the expression (\ref{16_11_2}) does not depend on~$k_2$.
It follows that the maximum of this expression is attained when $2n-2k_2=2n-(3/4)(n-k_2)$, i.e., when $k_2=(3/11)n$.
So the approximation ratio of Algorithm~${\cal A}_2$ is at most
$$\frac{2n-2(3/11)n}{n}=16/11.$$

It remains to note that finding a max-cardinality matching at Step~1 by the algorithm from \cite{Gabow} takes time $O(n^3)$, while Step~2 is performed in time $O(n^{2.5})$ by Theorem~\ref{th5_4}.
Thus, the running time of Algorithm~${\cal A}_2$ is $O(n^3)$.
The theorem is proved.
\hfill$\Box$
\end{proof}

\begin{remark}
The packing returned by Algorithm~${\cal A}_2$ contains either $1$- or $2$-unions.
So this packing is also an approximate solution of the version of $2$-BCPP where, given arbitrary $2$-BCs (not necessarily big), it is required to find a min-length packing, in which every cell of the strip contains at most two bars of different charts.
Note that, using a slightly modified algorithm based on the same ideas, easy to get a $19/14$--approximation of this version.
\end{remark}

\begin{remark}
It can be easily proved that, if we are given an oracle which returns an optimum solution to $2$-BCPP$|1$, e.g., if we can solve MaxATSP$(0,1)$ exactly, then the best of $P(M^*)$ and an optimum of $2$-BCPP$|1$ is a $4/3$--approximate solution to $2$-BCPP$|$big.
\end{remark}

\section{Conclusion}
We considered a problem in which it is necessary to pack $n$ two-bar charts (2-BCs) in a unit-height strip of minimum length. The problem is a generalization of the bin packing problem and 2-D vector packing problem. Earlier, we proposed an $O(n^2)$--time algorithm, which builds a packing of length at most $2\ OPT+1$ for arbitrary 2-BCs, where $OPT$ is the minimum length of the packing. Then, we proposed an $O(n^{3.5})$-- and $O(n^4)$--time packing algorithms based on the sequential matchings. Using only the first matching, one can construct a 3/2--approximate solution with time complexity $O(n^3)$ for big BCs (when each BC has at least one bar of height greater than 1/2) and with $O(n^{2.5})$ time complexity when additionally the BCs are non-increasing or non-decreasing.

This paper proposes two new packing algorithms based on matching and constructing an approximate solution to the MaxATSP(0,1). We prove that for packing arbitrary non-strictly big 2-BCs (the height of at least one bar of each BC is not less than 1/2), one algorithm constructs a 16/11--approximate solution with time complexity $O(n^3)$. If all 2-BCs are big (at least one bar of each BC has a height greater than 1/2) non-increasing or non-decreasing, then another algorithm constructs a 5/4--approximate solution with time complexity $O(n^{2.5})$.

We plan to conduct a numerical experiment to compare the solutions constructed by various approximation algorithms with the optimal solution yielded by the software package for BLP (for example, CPLEX)  \cite{Erzin20_2}. We are also planning to obtain a new accuracy estimate for the 2-BCPP problem with arbitrary 2-BCs.

\end{document}